\documentclass[11pt,fleqn]{article}
\usepackage{graphicx,pict2e,amssymb,amsmath,bm,fullpage}
\usepackage{ntheorem,delarray}
\usepackage{hyperref,srcltx}
\usepackage{makeidx}
\makeindex\nonstopmode
\long\def\killtext#1{}

\theoremseparator{.}
\newtheorem{theorem}{Theorem}[section]

\newtheorem{lemma}[theorem]{Lemma}

\newtheorem{corollary}[theorem]{Corollary}

\theorembodyfont{\rmfamily}

\newtheorem{exa}[theorem]{Example}

\newtheorem{problem}{Problem}[section]

\theoremindent.8cm \theorembodyfont{\small}

\newenvironment{proof}{\noindent{\bf Proof.}}{\hfill$\square$\medskip}

\newenvironment{proof*}[1]{\noindent{\bf Proof of #1.}}{\hfill$\square$\medskip}

\newcommand\numberthis{\addtocounter{equation}{1}\tag{\theequation}}

\def\Pr{{\sf P}}\def\E{{\sf E}}
\def\eps{\varepsilon}
\def\R{\sf R}
\def\N{{\cal N}}
\def\vol{{\sf vol}}
\def\iso{\ln(2)}
\def\k{\frac{1}{e}}

\begin{document}

\title{\Large\bf A Cubic Algorithm for Computing Gaussian Volume}

\author{Ben Cousins\thanks{Georgia Tech. Email:{\tt bcousins3@gatech.edu}}
\and
Santosh Vempala\thanks{Georgia Tech. Email: {\tt vempala@gatech.edu}}
}
\date{}
\maketitle
\thispagestyle{empty}

\begin{abstract}
We present randomized algorithms for sampling the standard Gaussian distribution restricted to a convex set and for estimating the Gaussian measure of a convex set,  in the general membership oracle model. The complexity of the integration algorithm is $O^*(n^3)$ while the complexity of the sampling algorithm is $O^*(n^{3})$ for the first sample and $O^*(n^2)$ for every subsequent sample. These bounds improve on the corresponding state-of-the-art by a factor of $n$. Our improvement comes from several aspects: better isoperimetry, smoother annealing, avoiding transformation to isotropic position and the use of the ``speedy walk" in the analysis. 
\end{abstract}

\newpage

\setcounter{page}{1}

\section{Introduction}

The study of high-dimensional sampling and integration has yielded many mathematical and algorithmic insights and tools that are now used in computer science, machine learning, statistics, operations research and other fields. Starting with the breakthrough of Dyer, Frieze and Kannan \cite{DyerFK89,DyerFK91}, who gave an $O^*(n^{23})$ randomized algorithm\footnote{The $O^*$ notation suppresses logarithmic factors and error parameters.} for estimating the volume of a convex body, a series of improvements and extensions were made over subsequent decades \cite{LS90, ApplegateK91, DyerF90, LS92, LS93, KLS97, LV2, LV07, LV06, CDV10}. For detailed accounts of these developments, see \cite{VemSurvey, Simonovits03}.

The current best complexity of integration is $O^*(n^4)$ for any logconcave function, in the general oracle model where the function of interest is accessible via an oracle that returns its value for any point in $\R^n$. Up to logarithmic factors, the state-of-the-art bounds for the important special cases of computing the volume (constant function on a convex body) and for Gaussian volume (Gaussian restricted to a convex body) are the same as for the general case. This complexity is achieved by the algorithm given in \cite{LV2003,LV2} (henceforth the ``LV algorithm") for convex bodies and extended to logconcave functions in \cite{LV06}. All known algorithms use the Markov chain Monte-Carlo approach, i.e., a reduction to sampling from a sequence of distributions. Thus, efficient sampling is at the core of fast volume algorithms and has also been intensively studied \cite{BubleyDJ98, KarzanovK91, KannanL96, LV3, KanNar2009}. The current best bound for sampling a convex body or logconcave density is $O^*(n^4)$ for the first sample and $O^*(n^3)$ for subsequent samples. In spite of being an active and well-known research topic, for the past decade, there has been no improvement in the worst-case complexity of sampling or volume computation. For Gaussian volume, there has been no improvement since the work of Kannan and Li \cite{KannanL96}. We note that while the upper bound has dropped remarkably from $O^*(n^{23})$ to $O^*(n^4)$, and lead to many techniques on the way,  it still appears to be outside the realm of being practical. 

Can the complexity of sampling and volume computation be further improved? There is an outstanding conjecture in this direction, namely the KLS hyperplane conjecture about the isoperimetric coefficient of logconcave functions \cite{KLS95}. Isoperimetric inequalities are a crucial ingredient in the analysis of random walks, which are the core of geometric sampling algorithms. Roughly speaking, such an inequality says that for any partition of space into two sets, the area of the separating surface is proportional to the volume (measure) of the smaller set. The ratio of the two cannot be too small for logconcave functions. The KLS conjecture says that if the corresponding logconcave density is in isotropic position (i.e., its covariance matrix is the identity), then this ratio is at least some absolute constant. The current best lower bound is $\Omega(n^{-1/3}\log^{-1/2} n)$ \cite{Eldan2013}. The recent work of Eldan \cite{Eldan2013}, Eldan and Klartag \cite{EldanK2011}, and earlier work of Ball \cite{Ball88, Ball09}, shows that the KLS conjecture is very closely (and quantitatively) connected to both the slicing conjecture (also called the hyperplane conjecture) and the thin-shell conjecture; these conjectures were formulated much earlier than KLS for entirely different reasons and are at the heart of asymptotic convex geometry.

The KLS conjecture, if true, implies that sampling an isotropic logconcave distribution using the ball walk from a warm start would take $O^*(n^2)$ steps instead of the current bound of $O^*(n^3)$. Since the LV algorithm uses only $O^*(n)$ samples, this suggests the possibility of reducing the complexity of integration by a factor of $n$ (from $O^*(n^4)$ to $O^*(n^3)$). However, besides proving the conjecture, there are other formidable hurdles. The conjecture and its implication for sampling hold only for {\em isotropic} or near-isotropic distributions, requiring that an algorithm maintain this property. This is possible with sampling, but expensive; each round of isotropization requires $\Omega(n)$ samples (and $O^*(n)$ samples suffice \cite{Rudelson1999, Adamczak2010}). The LV algorithm gets around this bottleneck by using isotropy only once, up front, and maintaining a more relaxed roundness condition during the course of the algorithm, thereby avoiding this overhead in every phase. One cannot hope to use this approach for a faster algorithm, as the KLS conjecture does not hold under the weaker roundness condition. A second important issue is that the ball walk, for which the KLS conjecture implies faster mixing, needs a warm start, i.e., the ratio of the starting density to the target density in any sampling phase is bounded for every point. In other words, the mixing time depends {\em polynomially} on how close the starting density is to the target. This is illustrated by the fact that starting the ball walk near a corner of a convex body could make the mixing very slow. As a result, recent improvements in isoperimetry towards the KLS conjecture do not have any direct impact on the complexity of sampling or volume computation. The LV algorithm resolves this issue by using {\em hit-and-run} \cite{Sm, LV3, LV06}, a random walk which mixes rapidly from any starting point. At the moment, there is no known analog of the KLS conjecture that would imply faster mixing for hit-and-run.

In this paper, we consider a well-motivated special case of logconcave integration, namely that of computing the Gaussian volume of a convex set. This is a natural setting of wide interest in probability, statistics and theoretical computer science. Integrating multivariate Gaussian distributions has been studied for many decades and has a wide variety of applications for statistics \cite{Martynov, Iyengar, kannanLi,Somerville, genzBretzBook}. It appears to be at least as important as the special case of computing the volume of convex body. As far as we know, there is no direct reduction between the two. We now state the Gaussian volume and sampling problems formally. We denote the Gaussian density function as $\gamma(x) = e^{-\|x\|^2/2}/(2\pi)^{n/2}$.

\begin{problem}\label{prob:volume}[Gaussian Volume]
Given a membership oracle for a convex set $K$ in $\R^n$ containing the unit ball $B_n$, and error parameter $\eps > 0$, give an algorithm that computes a number $V$ such that with probability at least $3/4$,
\[
(1-\eps) \int_K \gamma(x) \, dx \le V \le (1+\eps) \int_K \gamma(x) \, dx.
\]
%$V$ is a $(1+\eps)$ approximation to the Gaussian volume of $K$. 
\end{problem}

\begin{problem}\label{prob:sampling}[Gaussian Sampling]
Given a membership oracle for a convex set $K$ in $\R^n$ containing the unit ball $B_n$, and a paremeter $\eps > 0$, give an algorithm to generate a random point $x$ from $K$ whose density is within total variation distance $\eps$ of the standard Gaussian density restricted to $K$.
\end{problem}

Our main result is a cubic algorithm for Gaussian volume.

\begin{theorem}\label{thm:volume}
For any $\eps > 0$, $p>0$, and any convex set $K$ in $\R^n$ containing the unit ball, there is an algorithm that gives a $(1+\eps)$ approximation of the  
Gaussian volume of $K$ with probability $1-p$ and has complexity $O(n^3\log(n)\log^3(n/\eps)\log(1/p)/\eps^2)=O^*(n^3)$ in the membership oracle model.
\end{theorem}
The restriction on the Gaussian being centered at zero can be relaxed to being centered in a unit ball contained anywhere in the convex set. We also obtain an improved sampling algorithm.

\begin{theorem}\label{thm:sampling}
For any $\eps > 0$, $p>0$, and any convex set $K$ in $\R^n$ containing the unit ball, there is an algorithm that, with probability $1-p$, can generate random points from a density 
$\nu$ that is within total variation distance $\eps$ of the Gaussian density restricted to $K$. In the membership oracle model, the complexity of the first random point is $O(n^3\log(n)\log(n/\eps)\log(n/(\eps p)))=O^*(n^3)$ and is $O(n^2\log(1/\eps)\log(n/(\eps p)))=O^*(n^2)$ for subsequent random points; the set of random points will be $\eps$-independent.
\end{theorem}

Both results improve on the previous best complexity of $O^*(n^4)$, obtained via the general result for logconcave sampling and integration  \cite{LV06}. In earlier work, Kannan and Li \cite{KannanL96} considered the special case of sampling a Gaussian restricted to the positive orthant and obtained an $n^4$ algorithm. The main ingredient of the sampling theorem is the following mixing bound. For two probability distributions $P$ and $Q$ with state space $K$, let $M(P,Q)$ to denote the $M$-warmness between $P$ and $Q$, a measure of their proximity:
\begin{equation}\label{eqn:m-warm}
M(P,Q) = \sup_{S \subseteq K} \frac{P(S)}{Q(S)}.
\end{equation}

\begin{theorem}\label{thm:mixing}
Let $Q_0$ be a starting distribution and $Q$ be the target Gaussian density $\mathcal{N}(0, \sigma^2 I)$ restricted to $K \cap 4\sigma\sqrt{n}B_n$ for $\sigma \le 1$ and a convex set $K$ containing the unit ball. For any $\eps > 0$, the lazy Metropolis ball walk with $\delta$-steps, for $\delta = \sigma/(4096\sqrt{n})$, starting from $Q_0$, satisfies $d_{tv}(Q_t,Q) \le \eps$ with probability $1-p$ after
\[
t \ge C  \cdot M(Q_0,Q)\cdot n^2\log(1/\eps)\log(1/p)
\]
steps for an absolute constant $C$.
\end{theorem}

\section{Overview of results and techniques}

We will use random walks for sampling. For technical reasons, one usually performs a {\em lazy} version where in each step, with probability half, the walk stays put; alternatively, the number of steps of the walk is chosen randomly in advance. The {\em Metropolis Ball Walk with target density $f$} is outlined in Figure \ref{algo:sampler}.

\begin{figure}[h]
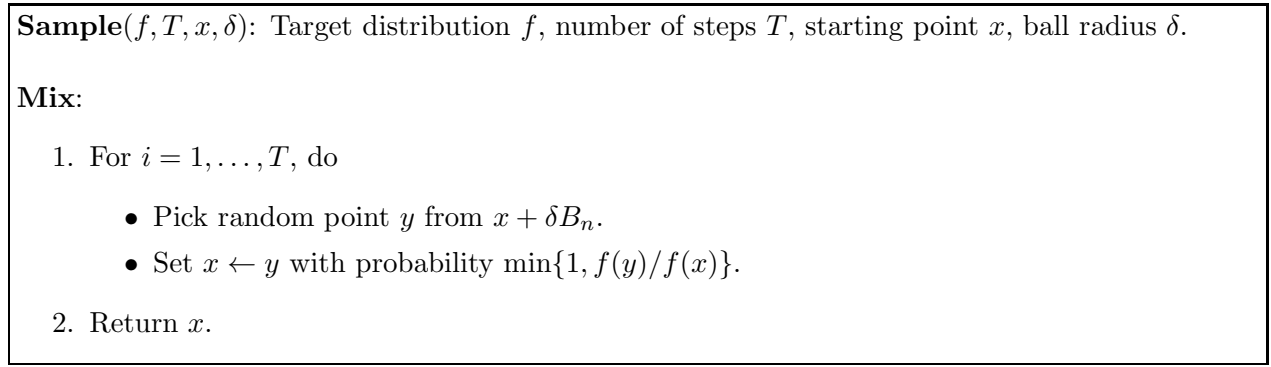

\fbox{\parbox{\textwidth}{
{\bf Sample}$(f,T,x,\delta)$: Target distribution $f$, number of steps $T$, starting point $x$, ball radius $\delta$. \\

{\bf Mix}: 
\begin{enumerate}
\item For $i=1, \ldots , T$, do
\begin{itemize}
\item Pick random point $y$ from $x + \delta B_n$.
\item Set $x \leftarrow y$ with probability $\min\{1, f(y)/f(x)\}$.
\end{itemize}
\item Return $x$.
\end{enumerate}
}}
\caption{Ball walk sampler}
\label{algo:sampler}
\end{figure}

%\medskip
%\noindent
%At current point $x$:
%\begin{enumerate}
%\item Pick a random point $y$ from $B(x,\delta)$.
%\item If $y \in K$, go to $y$ with probability $\min\left\{ 1, \frac{f(y)}{f(x)}\right\}$.
%\end{enumerate}

This will be the core of our sampling algorithm. 

\subsection{Outline of sampling algorithm and analysis}

To show the random walk quickly reaches its stationary distribution, we will use the standard method of bounding the conductance.   
For the ball walk, this runs into a hurdle, namely, the local conductance of points near sharp corners of the body can be arbitrarily small, so the walk can get stuck and waste a large number of steps. To avoid this, we could start the walk from a random point chosen from a distribution sufficiently close to the target distribution. But how to generate random points from such a starting distribution? We do this by considering a sequence of distributions, each providing a warm start for the next. The very first distribution is chosen to be a highly concentrated Gaussian so that it almost entirely lies inside the unit ball (inside $K$). Thus sampling from the initial distribution is easy. Each successive Gaussian is ``flatter" with the final one being the standard Gaussian. 
We will see that a sequence of $O^*(n)$ Gaussians suffices in the sense that a random point from one provides a warm start for sampling from the next. 

The next challenge is to show that, from a warm start, the expected number of steps to converge to the stationary distribution is only $O^*(n^2)$. This is usually done by bounding the conductance of the the Markov chain. 
The conductance, $\phi$, of a Markov chain with state space $K$ and next-step distribution $P_x$ is defined as:
\[
\phi = \min_{S \subset K} \frac{\int_S P_x(K\setminus S) \, dQ(x)}{\min Q(S), Q(K\setminus S)}.
\]
Unfortunately, for the ball walk, this can be arbitrarily small, e.g., for points near corners (but also for points in the interior). To utilize the warm start, we use an idea from \cite{KLS97}, namely the {\em speedy} walk. We emphasize that the speedy walk cannot be implemented efficiently and is only a tool for analysis. It is defined as follows.

\noindent
At current point $x$:
\begin{enumerate}
\item Pick random point $y$ from $K \cap x + \delta B_n$.
\item Go to $y$ with probability $\min \{1, f(y)/f(x)\}$.
\end{enumerate}

To capture the stationary distribution of the speedy walk with a Metropolis filter we need another parameter. The {\em local conductance} at $x$ for the speedy walk, without a filter, is defined as follows:
\[
\ell(x) = \frac{\vol(K\cap x+\delta B_n)}{\vol(\delta B_n)}.
\]
The following fact is now easy to verify.
\begin{lemma}\label{lem:speedystationary}
The stationary distribution of the speedy walk with a Metropolis filter applied with a function $f$ has density proportional to $\ell(x)f(x)$.
\end{lemma}
For the speedy walk, we can show that the conductance is $\Omega(1/n)$, and so the total number of steps needed is only $O^*(n^2)$. This is a factor $n$ faster than previous best bounds. We do this by establishing a stronger (and nearly optimal) isoperimetric inequality.

As noted, the speedy walk cannot actually be implemented efficiently. To bound the Metropolis ball walk, we can view it as an interleaving of a speedy walk with wasted steps. Let the Markov chain for the original walk is $w_0, w_1, \ldots, w_i, \ldots,$. The subsequence $w_{i_1}, w_{i_2}, \ldots,$ where we record $x$ if the point $y$ chosen by the Metropolis ball walk is in $K$, corresponds to the speedy walk. We then need to estimate the number of wasted steps from a warm start. We will show that this is at most a constant factor higher than the number of proper steps. The key ingredient of this analysis is the (known) fact that for a body containing the unit ball average local conductance is high for ball radius $\delta = O(1/\sqrt{n})$. Even within the speedy walk, there are ``null" steps due to the Metropolis filter. However, by restricting the walk to a large ball, we ensure that the probability of rejection by the filter is bounded by a constant, and therefore the number of wasted steps within the speedy walk is at most a constant fraction of all steps.

\subsection{Outline of volume algorithm and analysis}

The Gaussian volume algorithm runs on top of the sampling algorithm. As in sampling, it starts with a Gaussian that is highly focused inside the unit ball contained in $K$, then flattens out the Gaussian in a series of phases $1,2 \ldots, k$. The initial density $f_0$ is chosen so that its measure outside $K$ is negligible. In each phase $i$, we estimate the ratio of the measure of $K$ according to $f_i$ to the measure according to the previous density $f_{i-1}$. Unlike sampling, where we just maintain a single random point and adjust the target density (``temperature") periodically, for the volume algorithm we need multiple samples at each temperature. Fortunately, this can be done using only $O^*(\sqrt{n})$ phases with $O^*(\sqrt{n})$ samples in each phase. The sampling algorithm needs $O^*(n)$ phases but only maintains one random point; so $\sqrt{n}$ pure sampling phases (with no computation of ratios) are done between consecutive phases of the volume algorithm. 

The phases of the volume algorithm are chosen so that the variance of the ratio estimator is small. The main part of the analysis is showing that this can be done with $O^*(\sqrt{n})$ phases and $O^*(\sqrt{n})$ samples per phase. The idea of the algorithm and the tools used in the analysis are similar to the LV algorithm, with two significant departures. The first is that we only use Gaussian densities. This is important for maintaining the mixing time. The second is that 
unlike all previous integration/volume algorithms, we do not need to make the body/distribution isotropic or otherwise ``round" the body. This significantly simplifies the algorithm and is also important for our improved time complexity.

\subsection{Preliminaries}
A function $f : \mathbb{R}^n \rightarrow \mathbb{R}^+$ is \emph{logconcave} if it has convex support and the logarithm of $f$, wherever $f$ is non-zero, is concave. Equivalently, $f$ is logconcave if for any $x,y \in \mathbb{R}^n$ and any $\lambda \in [0,1]$, 
\[
f(\lambda x + (1  - \lambda) y ) \geq f(x)^\lambda f(y)^{1-\lambda}
\]

Let $\gamma:\R^n \rightarrow \R_+$ be the density of the standard Gaussian $\N(0,I)$.  

For two probability distributions $P$ and $Q$ with state space $K$, we will use $M(P,Q)$ to denote the $M$-warmness between $P$ and $Q$ as defined in \ref{eqn:m-warm}
%\[
%M(P,Q) = \sup_{S \subseteq K} \frac{P(S)}{Q(S)},
%\]
and $d_{tv}(P,Q)$ to denote the total variation distance between $P$ and $Q$:
\[
d_{tv}(P,Q) = \sup_{S \subseteq K} | P(S) - Q(S) | .
\]

For a nonnegative function $f:\R^n \rightarrow \R_+$, we define the $f$-distance between two points $u,v \in \R^n$ as
\[
d_f(u,v) = \frac{|f(u)-f(v)|}{\max\{f(u), f(v)\}}.
\]

\section{Algorithm}

We now describe the main Gaussian volume algorithm. In the description below, we assume that the convex set $K$ contains a unit ball centered at the origin (which is also the mean of the Gaussian). The heart of the algorithm is the ball walk, described in Figure \ref{algo:sampler}, which takes in a current point $x$ from some distribution and lets it mix for an appropriate number of steps. In our Gaussian volume algorithm (Figure \ref{algo:volume}), the number of steps is chosen so that the point becomes sufficiently close to the target distribution $f$ (i.e. total variation distance at most $\nu$) with a $\log(1/p)$ overhead, where $p=1/(20\cdot \#phases)$ for an overall sampling failure probabilty of $1/20$.

\begin{figure}[h]
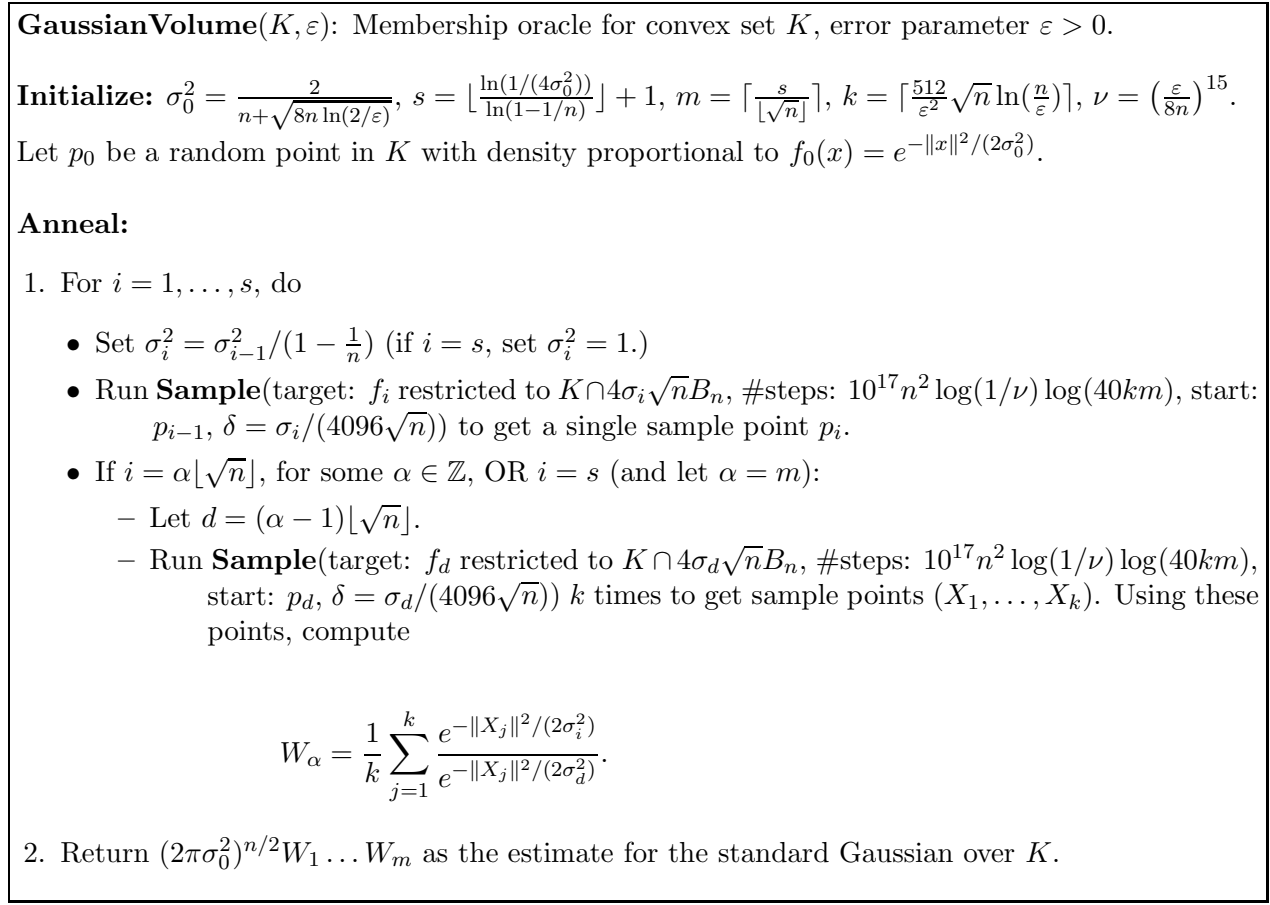

\fbox{\parbox{\textwidth}{

{\bf GaussianVolume}$(K,\eps)$: Membership oracle for convex set $K$, error parameter $\eps > 0$.\\

\noindent
{\bf Initialize:} $\sigma_0^2 = \frac{2}{n+\sqrt{8n\ln(2/\eps)}}$, $s = \lfloor \frac{\ln(1/(4\sigma_0^2))}{\ln(1-1/n)} \rfloor+1$, $m = \lceil \frac{s}{\lfloor \sqrt{n} \rfloor }\rceil$, $k=\lceil \frac{512}{\eps^2}\sqrt{n} \ln(\frac{n}{\eps}) \rceil$, $\nu = \big (\frac{\eps}{8n}\big )^{15}$. \\
Let $p_0$ be a random point in $K$ with density proportional to $f_0(x) = e^{-\|x\|^2/(2\sigma_0^2)}$.\\

{\bf Anneal:}

\noindent
\begin{enumerate}
\setlength{\itemindent}{-1em}
\item For $i=1,  \ldots  , s$, do
\begin{itemize}
\setlength{\itemindent}{-2em}
\item Set $\sigma_i^2 = \sigma_{i-1}^2/(1-\frac{1}{n})$ (if $i=s$, set $\sigma_i^2=1$.)
\item Run {\bf Sample}(target: $f_i$ restricted to $K \cap 4 \sigma_i \sqrt{n} B_n$, \#steps: $10^{17}n^2\log(1/\nu)\log(40km)$, start: $p_{i-1}$, $\delta = \sigma_i/(4096\sqrt{n})$) to get a single sample point $p_i$.
\item If $i = \alpha \lfloor \sqrt{n} \rfloor$, for some $\alpha \in \mathbb{Z}$, OR $i=s$ (and let $\alpha=m$):
\begin{itemize}
\setlength{\itemindent}{-2em}
\item Let $d=(\alpha-1)\lfloor \sqrt{n} \rfloor$.
\item Run {\bf Sample}(target: $f_d$ restricted to $K \cap 4 \sigma_d \sqrt{n} B_n$, \#steps: $10^{17}n^2\log(1/\nu)\log(40km)$, start: $p_d$, $\delta = \sigma_d/(4096\sqrt{n})$) $k$ times to get sample points $(X_1, \ldots, X_k)$. Using these points, compute 

\[
W_{\alpha} = \frac{1}{k} \sum_{j=1}^k \frac{e^{ - \|X_j\|^2/(2\sigma_i^2)}}{e^{ - \|X_j\|^2/(2\sigma_d^2)}}.
\]
\end{itemize}
\end{itemize}
\item Return $(2 \pi \sigma_0^2)^{n/2}W_1 \ldots W_{m}$ as the estimate for the standard Gaussian over $K$.
\end{enumerate}
}}
\caption{Gaussian volume algorithm}
\label{algo:volume}
\end{figure}

In Figure \ref{algo:volume}, we describe how to use the sampler described in Figure \ref{algo:sampler} to obtain the Gaussian volume of $K$. Define $f_i(x) = e^{-\|x\|^2/(2\sigma_i^2)}$. We select $\sigma_0^2$ such that
\[
\int_K f_0(x) dx \geq (1-\frac{\eps}{2})\int_{\R^n}f_0(x)dx,
\]
and therefore the volume of $f_0(x)$ inside $K$ is approximately $(2\pi\sigma_0^2)^{n/2}$. We slowly approach the standard Gaussian by continually dividing the variance by $1-1/n$, which maintains an $M$-warmness at most $\sqrt{e}$ between consecutive phases. We then mark every $\lfloor \sqrt{n}\rfloor$-th phase as a volume phase, and compute the ratio of its integral over $K$ with the integral of the previous volume phase over $K$. Our estimate for the standard Gaussian volume over $K$  is then (where $\sigma_s=1$)
\[
(2\pi\sigma_0^2)^{n/2}\cdot \frac{\int_K f_{\lfloor \sqrt{n} \rfloor}(x)dx}{\int_K f_0(x)dx} \cdot \frac{\int_K f_{2\lfloor \sqrt{n} \rfloor}(x)dx}{\int_K f_{\lfloor \sqrt{n} \rfloor}(x)dx} \ldots \frac{\int_K f_s(x)dx}{\int_K f_{(m-1)\lfloor \sqrt{n}\rfloor}(x)dx}.
\]

To see that the complexity of the algorithm is $O^*(n^3)$, we can consider it as two components: sampling phases and volume phases (and bound each complexity separately). 

\begin{enumerate}
\item Sampling Phases:  $O^*(n)$ phases $\times$ $O^*(1)$ samples per phases $\times$ $O^*(n^2)$ steps per sample $=O^*(n^3)$.
\item Volume Phases: $O^*(\sqrt{n})$ phases $\times$ $O^*(\sqrt{n})$ samples per phases $\times$ $O^*(n^2)$ steps per sample $=O^*(n^3)$.
\end{enumerate}

\section{Isoperimetry}

The following theorem is due to Brascamp and Lieb. 
\begin{theorem}\cite{BL76}\label{thm:Brascamp-Lieb}
Let $\gamma:\R^n\rightarrow\R_+$ be the standard Gaussian density in $\R^n$. Let $f:\R^n\rightarrow \R_+$ be any logconcave function. Define the density function $h$ in $\R^n$ as follows:
\[
h(x) = \frac{f(x)\gamma(x)}{\int_{\R^n} f(y)\gamma(y) \, dy}.
\]
Fix a unit vector $v \in \R^n$ ,  let $\mu = \E_h(x)$. Then, for any $\alpha \ge 1$,
\[
\E_h(|v^T(x -\mu)|^\alpha) \le \E_{\gamma}(|x_1|^\alpha).
\]
\end{theorem}

We have the following concentration bound.

\begin{corollary}\label{cor:iso-conc-bound}
For $h$ as defined in Theorem \ref{thm:Brascamp-Lieb}, and any $t \ge 1$,
\[
\Pr_h(\|x - \mu\|^2  \ge n+ct\sqrt{n}) \le e^{-t^2}
\]
for an absolute constant $c$.
\end{corollary}

%The following theorem is due to Paouris \cite{Paouris2006}.

%\begin{theorem}\label{thm:logconcave-bound}\cite{Paouris2006}
%Let $X$ be drawn from an isotropic logconcave measure. Then, for any $t \geq 1$, 
%\[
%P(\|X\| \geq c t \sqrt{n}) \leq e^{-t\sqrt{n}}
%\]
%where $c$ is an absolute constant.
%\end{theorem}
The next lemma about one-dimensional isoperimetry is from \cite{KLS95}

\begin{lemma}\label{lem:1d-iso}\cite{KLS95}
For any one-dimensional isotropic logconcave function $f$, and any partition $S_1, S_2, S_3$ of the real line, 
\[
\pi_f(S_3) \ge \iso \, d(S_1,S_2) \pi_f(S_1)\pi_f(S_2).
\] 
\end{lemma}

\iffalse
\begin{proof}
By Theorem 5.1 of \cite{KLS95}, 
\[
\pi_f(S_3) \ge \frac{\ln 2}{\E_f(\|x\|)}d(S_1,S_2) \pi_f(S_1)\pi_f(S_2).
\] 
To prove the lemma we note that 
\[
\E_f(\|x\|) \le \sqrt{\E_f(\|x\|^2)} = 1
\]
and 
\[
\pi_f(S_3) \ge c \pi_f(S_1)\pi_f(S_2) \ge c \pi_f(S_1)(1-\pi_f(S_1) - \pi_f(S_3))
\]
which, assuming $\pi_f(S_1) \le \pi_f(S_2)$ and therefore $\pi_f(S_1) \le 1/2$, gives
\[
\pi_f(S_3) \ge \frac{c}{2+c} \min\{ \pi_f(S_1), \pi_f(S_2)\}.
\]
\end{proof}
\fi

\begin{theorem}\label{thm:iso}
Let $\pi$ be the Gaussian distribution $N(0,\sigma^2 I_n)$ with density function $\gamma$ restricted by a logconcave function $f:\R^n \rightarrow \R_+$, i.e., $\pi$ has
density $d\pi(x)$ proportional to $h(x)=f(x)d\gamma(x)$. Let $S_1,S_2,S_3$ partition $\R^n$ such that
for any $u \in S_1, v\in S_2$, either $\|u-v\| \ge d/\iso$ or $d_h(u,v) \ge 4d\sqrt{n}$. Then, 
\[
\pi(S_3) \ge \frac{d}{\sigma}\pi(S_1)\pi(S_2).
\]
\end{theorem}

\begin{proof}
We prove the theorem for the case $\sigma=1$, then note that by applying the scaling $x = y/\sigma$, we get the general case.

Our main tool, as in previous work, is the Localization Lemma of Lov\'asz and Simonovits \cite{LS93}. Suppose the conclusion is false. Define $h(x)=f(x)\gamma(x)$.
 Then there exists a partition $S_1, S_2, S_3$ for which, for some positive real number $A$, 
\begin{align*}
\int_{S_1} h(x)\, dx &=  A\int_{\R^n}h(x)\, dx \\
\mbox{ and } \int_{S_3} h(x)\, dx &<  d A \int_{S_2} h(x)\, dx.
\end{align*}
By the localization lemma, there must be a ``needle'' given by $a,b\in \R^n$ and a nonnegative linear function $l: [0,1] \rightarrow \R_+$ for which,
\begin{align*}
\int_{(1-t)a+tb \in S_1 \cap [0,1]} h((1-t)a+tb) l(t)^{n-1}\, dt &= A \int_{(1-t)a+tb \in [0,1]} h((1-t)a+tb) l(t)^{n-1}\, dt\\
\int_{(1-t)a+tb \in S_3 \cap [0,1]} h((1-t)a+tb) l(t)^{n-1}\, dt &< d A\int_{(1-t)a+tb \in S_2 \cap [0,1]} h((1-t)a+tb)l(t)^{n-1}\, dt.
\end{align*}
By a standard combinatorial argument, we can assume that $Z_i = \{t: (1-t)a+tb \in S_i\}$ are intervals that partition $[a,b]$.
Thus, to reach a contradiction, it suffices to prove that for a one-dimensional logconcave function $h(t)=f((1-t)a+tb)\gamma((1-t)a+tb)$ with support $[a,b] \subset \R$ and $a\le u\le v \le b$, the
following statements hold:
\begin{eqnarray}
\int_{a}^b h(t)l(t)^{n-1}\, dt \int_u^v h(t)l(t)^{n-1}\, dt &\ge& \frac{d_h(u,v)}{4\sqrt{n}} \int_a^u  h(t)l(t)^{n-1}\, dt  \int_v^b  h(t)l(t)^{n-1} \label{1d-1}\\
\int_a^b h(t)l(t)^{n-1}\, dt \int_u^v h(t)l(t)^{n-1}\, dt &\ge& \iso \|u-v\| \int_a^u  h(t)l(t)^{n-1}\, dt  \int_v^b  h(t)l(t)^{n-1} \label{1d-2}.
\end{eqnarray}
The first inequality (\ref{1d-1}) follows directly from Lemma 3.8 in \cite{KLS97}. 
To see the second inequality (\ref{1d-2}), we first note that by applying Theorem \ref{thm:Brascamp-Lieb}, with $\alpha=2$, 
we have that the variance of 
the distribution proportional to $h(t)l(t)^{n-1}$ is at most $1$. This is because $h(t)l(t)^{n-1} = (f((1-t)a+tb)l(t)^{n-1})\gamma((1-t)a+tb)$ and the $f((1-t)a+tb)l(t)^{n-1}$ is itself a logconcave function. Now, we note that by scaling down to increase the variance to exactly $1$, the isoperimetric coefficient can only go down. Hence, the second inequality is implied by Lemma \ref{lem:1d-iso}.
\end{proof}

\section{Sampling}
The analysis of the sampling algorithm is divided into several parts: bounding the conductance of the speedy walk, the mixing time of the Metropolis ball walk from a warm start, bounding the warmth of the distribution from one phase to the next, and finally the complexity of sampling.

\subsection{Conductance}

\begin{lemma}\label{lem:overlap}
Let $S, \bar{S}$ be a partition of a convex body $K$, and $u \in S, v\in \bar{S}$ be such that $\|u - v\| < \delta/\sqrt{n}$. Then,
\[
\Pr_u(\bar{S}) + \Pr_v(S) \ge \frac{1}{e(e+1)} \min \{\ell(u), \ell(v)\}.
\]
\end{lemma}
\begin{proof}
This lemma is essentially based on Lemma 3.6 in \cite{KLS97}. There the speedy walk makes a uniform step. Here we apply a Metropolis filter. However, since we restrict to a body of radius $4\sigma\sqrt{n}$ and $\delta \leq \sigma/8\sqrt{n}$, we have that the acceptance probability of the filter is at least $1/e$. 
\end{proof}

\begin{lemma}\label{lem:lulv}
Let $S, \bar{S}$ be a partition of a convex body $K$, and $u \in S, v\in \bar{S}$ be such that $\|u - v\| < \delta/\sqrt{n}$ and $\Pr_u(\bar{S}) \le \ell(u)/20e, \Pr_v(S) \le \ell(v)/20e$. Then, 
\[
\frac{|\ell(u)\gamma(u) - \ell(v)\gamma(v)|}{\max \{\ell(u)\gamma(u), \ell(v)\gamma(v)\}} \ge \frac{1}{4}.
\]
\end{lemma}

\begin{proof}
Assume that $\ell(u) \ge \ell(v)$. Then by Lemma \ref{lem:overlap},
\[
\frac{1}{e(e+1)}\ell(v) \le \Pr_u(\bar{S}) + \Pr_v(S) \le \frac{1}{20e}(\ell(u) + \ell(v))
\]
which implies that 
\[
\ell(v) \le \frac{e+1}{20-(e+1)}\ell(u).
\]
Therefore,
\begin{align*}
\frac{|\ell(u)\gamma(u) - \ell(v)\gamma(v)|}{\ell(u)\gamma(u)} &= 1 - \frac{\ell(v)}{\ell(u)}\frac{\gamma(v)}{\gamma(u)} \\ 
&\ge 1 - \frac{e(e+1)}{20-(e+1)} \ge \frac{1}{4}.
\end{align*}

\end{proof}

\begin{theorem}\label{thm:speedyconductance}
Let $K$ be a convex body such that $B_n \subseteq K \subseteq 4\sigma \sqrt{n}B_n$.
The conductance of speedy walk applied to $K$ with Gaussian density ${\cal N}(0,\sigma^2 I)$ and  $\delta \le \sigma/8\sqrt{n}$ steps is $\Omega(\frac{\delta}{\sigma\sqrt{n}})$. 
\end{theorem}
\begin{proof}
Let $S \subset K$ be an arbitrary measurable subset of $K$ and consider the following partition of $K$.
\[
S_1 = \{x \in S \, :\, P_x(\bar{S}) < \frac{\ell(x)}{20e}\}\\
S_2 = \{x \in \bar{S}\, :\, P_x(S) < \frac{\ell(x)}{20e}\}\\
S_3 = K\setminus S_1\setminus S_2.
\]
Let $h(x)=\ell(x)\gamma(x)$.
We claim that for any $u \in S_1, v\in S_2$, we have either $\|u-v\| \ge \delta/\sqrt{n}$ or 
\[
d_h(u,v) = \frac{|h(u)-h(v)|}{\max \{h(u), h(v)\}} \ge \frac{1}{4}.
\]
This follows directly from Lemma \ref{lem:lulv}.

We can also assume that $\pi(S_1) \ge \pi(S)/2$ and $\pi(S_2) \ge \pi(\bar{S})/2$. If not, ergodic flow (probability of going from $S$ to $\bar{S}$) can be bounded as follows.
\begin{eqnarray*}
\Phi(S) &=& \frac{1}{2}(P(S,\bar{S})+P(\bar{S},S))\\ 
&\ge& \frac{1}{2}\int_{S_3}\k\frac{\ell(x)}{20e}\gamma(x)\, dx 
\end{eqnarray*}
Here $\k$ comes from the minimum probability of the Metropolis filter. We obtain this by noting that the minimum value of 
\[
\frac{f(x)}{f(u)} \ge \exp(-\frac{2\|x-u\|\|u\|+\|x-u\|^2}{2\sigma^2}) \ge e^{-\delta\|u\| - \frac{\delta^2}{2\sigma^2}} \ge \frac{1}{e}.
\]
Here we have used the fact that $\delta \le \sigma/8\sqrt{n}$ and from $\|u\| \le 4\sigma\sqrt{n}$.

Now we can apply Theorem \ref{thm:iso}, with 
\[
d=\min \left\{ \iso \frac{\delta}{\sqrt{n}}, \frac{1}{16\sqrt{n}} \right\}
\]
to the partition $S_1, S_2, S_3$ to get
\[
\int_{\R^n} \ell(x)\gamma(x) \, dx \int_{S_3}\ell(x)\gamma(x)\, dx \ge \frac{d}{\sigma}  \int_{S_1} \ell(x) \gamma(x)\, dx \int_{S_2} \ell(x)\gamma(x)\, dx.
\]
Using this, we get
\[
\Phi(S) \ge  \frac{1}{10^{4}}\frac{\delta}{\sigma\sqrt{n}}\min\{\pi_\ell(S), \pi_\ell(\bar{S})\}
\]
The conductance is thus $\Omega(\delta/\sigma\sqrt{n})$ as claimed.
\end{proof}

\subsection{Mixing from a warm start}

To show that the conductance will imply a rate of convergence of the speedy walk, we will use a result of \cite{LS93} to bound the total variation distance between the current distribution and the target distribution.

\begin{theorem}\label{thm:mixingtime}
Let $Q_t$ be the distribution after $t$ steps of a lazy Markov chain and $Q$ be its stationary distribution. Then,
\[
d_{tv}(Q_t, Q) \leq M(Q_0,Q) (1- \frac{\phi^2}{2})^t
\]
\end{theorem}

\noindent
\textbf{Remark:} If we want $d_{tv}(Q_t,Q)\leq \eps$, then we can select $t=\lceil 2 \phi^{-2} \ln (\frac{M(Q_0,Q)}{\eps}) \rceil$. \\

To bound the number of wasted steps, we will need the following average local conductance of the ball walk with a Metropolis filter:
\[
\lambda(f) = \frac{\int_{K} \ell(x) f(x)\, dx}{\int_K f(x)\, dx}.
\]
We say that a density function $f:\R^n \rightarrow \R_+$ is {\em $a$-rounded} if any level set $L$ contains a ball of radius $a \cdot \mu_f(L)$. 

We will show that $\lambda$ is large for "well-rounded" distributions.
\begin{lemma} \label{lem:alc}
For any $a$-rounded logconcave density function $f$ in $\R^n$, 
\[
\lambda(f) \ge 1 - 32 \frac{\delta^{1/2}n^{1/4}}{a^{1/2}}.
\]
\end{lemma}
\begin{proof}
Define $\hat{f}$ as the following smoothened version of $f$, obtained by convolving $f$ with a ball of radius $\delta$. 
Let $D$ be a convex subset of $\delta B_n$ of half its volume.
\[
\hat{f}(x) = \min_{D} \frac{\int_{y \in x+D} f(y)\, dy}{\vol(D)}.
\]
Now Lemma 6.3 from \cite{Lovasz2007} shows that
\[
\int_K \hat{f}(x)\, dx \ge 1 - 32 \frac{\delta^{1/2}n^{1/4}}{a^{1/2}}.
\]
To complete the proof, we observe that for any point $x$,
\[
\ell(x)f(x) \ge \hat{f}(x).
\]
To see this, note that 
\begin{eqnarray*}
\ell(x)f(x) &=& f(x) \frac{\int_{x+\delta B_n} 1\, dy}{\vol(\delta B_n)}\\
&\ge& \frac{\int_{y \in x+\delta B_n: f(y) \le f(x)} f(y)\, dy}{\int_{y \in x+\delta B_n:f(y) \le f(x)} 1 \, dy}\\
&\ge& \hat{f}(x).
\end{eqnarray*}
\end{proof}

We will also show the following bound on the roundness of our distribution.
\begin{lemma}\label{lem:sigma-round}
The Gaussian $N(0,\sigma^2 I_n)$ for $\sigma \le 1$, restricted to $K$ containing a unit ball centered at zero is $\sigma$-rounded.
\end{lemma}
\begin{proof}
%The level sets of the distribution are balls restricted to $K$. Thus, every level set either contains a unit ball or is itself a ball of radius $r$ smaller than $1$. In the latter case, we note that,
%\[
%\mu(L) \le \frac{\int_{r B_n} d\gamma(x)}{\int_{B_n} d\gamma(x)}.
%\]
The level sets of the distribution are balls restricted to $K$. For the distribution to be $\sigma$-rounded, we need that a level set of measure $t$ contains a ball of radius $t \sigma$. Consider the following function of $t$:
\[
g(t) = \frac{\int_0^t e^{-x^2/(2\sigma^2)}x^{n-1}dx}{ \int_0^\sigma e^{-x^2/(2\sigma^2)}x^{n-1}dx}.
\]
Consider the second derivative of $g(t)$:
\begin{align*}
g''(t) = (\sigma^2(n-1)-t^2)\cdot\frac{t^{n-2}e^{-t^2/(2\sigma^2)}}{ \int_0^\sigma e^{-x^2/(2\sigma^2)}x^{n-1}dx}
\end{align*}

For $g''(t)$ to be non-negative, we need $(\sigma^2(n-1)-t^2)\geq 0$, which it is for $n \geq 2, t\in [0,\sigma]$. Since $g(0)=0, g(\sigma)=1$, and the second derivative is non-negative, we then have that $g(\sigma t)\leq t$ for $t \in [0,1]$, which proves the lemma.
\end{proof}

\begin{corollary}
If $\delta \leq \sigma/(4096\sqrt{n})$, then the average local conductance, $\lambda(f)$, for the density function $f$ proportional to the Gaussian $N(0,\sigma^2 I_n)$ restricted to $K$, is at least $1/2$.
\end{corollary}
\begin{proof}
Using Lemma \ref{lem:alc} and Lemma \ref{lem:sigma-round}, we have that
\[
\lambda(f) \geq 1-32\frac{\sigma^{1/2}n^{1/4}}{C^{1/2}n^{1/4}\sigma^{1/2}} = 1-32\frac{1}{C^{1/2}}.
\]

We can select $C=4096$, which proves the corollary.
\end{proof}
%\begin{proof}

%For $\delta \leq 1/(C\sqrt{n})$, at least $1-2/C$ of $K$ has $\ell(x)$ at least $3/4$. Then, 
%\[
%\E(\ell(x)) \geq (1-\frac{2}{C}) \cdot \frac{3}{4} = \frac{3}{4} - \frac{3}{2C}.
%\]

%Then, we can select $C\geq 6$, which proves the corollary.
%\end{proof}

%\begin{lemma}
%If $\delta \leq 1/(4\sqrt{n})$, then the average local conductance, $\E(\ell(x))$,  is at least 3/(8e).
%\end{lemma}
%\begin{proof}
%If we set $\delta \leq \eps / 2\sqrt{n}$, then we get that that at least $(1-e)$ fraction of the points in $x \in K$ have  $\lambda (x,\delta) > 3/4$. Let $\eps=1/2$. Then,
%\[
%\E(\lambda(x,\delta)) \geq \frac{1}{2}\cdot \frac{3}{4} = \frac{3}{8}.
%\]
%Since $K \subseteq C \sqrt{n} B_n$, the Metropolis filter will have acceptance at least $1/e$, and therefore 
%\[
%\E(\ell(x)) \geq  \frac{3}{8e}.
%\]

%\end{proof}
\begin{lemma}\label{lem:ballsteps}
If the average local conductance is at least $\lambda$, $M(Q_0,Q) \leq M$, and the speedy walk takes $t$ steps, then the expected number of steps of the corresponding ball walk is at most
$
Mt/\lambda.
$
\end{lemma}

\begin{proof}
Since $M(Q_0,Q)\leq M$, we have that for all $S \subseteq K$,
\[
Q_0(S) \leq M Q(S),
\]

and by induction on $i$, we get that
\[
Q_i(S) = \int_K P_x(S) dQ_{i-1}(x) \leq M\int_K P_x(S) dQ(x) = MQ(S).
\]

For any point $x$, the expected number of steps until a proper step is made is $1/\ell(x)$. So, given a point from $Q_i$, the expected number of steps to obtain a point from $Q_{i+1}$ is 
\[
\int_K \frac{1}{\ell(x)}dQ_i(x) \leq M\int_K \frac{1}{\ell(x)}dQ(x) \leq \frac{M}{\lambda} \int_K d\hat Q (x) = \frac{M}{\lambda},
\]

where $\hat Q$ is the corresponding distribution for the ball walk with a Metropolis filter (i.e., with stationary distribution proportional to $f(x)$). If the speedy walk took $t$ steps, then by linearity of expectation, the expected number of steps for the ball walk is at most $Mt/\lambda$.
\end{proof}

%\begin{remark}
%The above analysis was done for the standard Gaussian, but it extends to any Gaussian if $\delta = \sigma/(C\sqrt{n})$.
%\end{remark}

We can now prove Theorem \ref{thm:mixing}.\\

\begin{proof}(of Theorem \ref{thm:mixing}.)
By Theorem \ref{thm:speedyconductance}, we have that the conductance of the speedy walk is $\Omega(\delta/\sigma\sqrt{n}) = \Omega(1/n)$ if $\delta \le 1/(4096\sqrt{n})$. Then, from Theorem \ref{thm:mixingtime}, we have that the speedy walk will take $C' n^2 \log (M(Q_0,Q) \cdot 1/\eps)$ steps until the variation distance is at most $\eps$. 

From Lemma \ref{lem:ballsteps}, we see that the expected number of steps for the corresponding ball walk is at most $2 M(Q_0,Q) C'n^2\log(M(Q_0,Q) \cdot 1/\eps) = C_1 M(Q_0,Q) n^2 \log(1/\eps)$. To obtain a guaranteed number of steps with some failure probabilty $p$, if we run the ball walk for twice the expected number of steps, we will achieve the required number of speedy steps with probability $1/2$ by Markov's inequality. We can then repeat $\log(1/p)/\log(2)$ times, but concatenated together as one long chain of steps (since using a mixed point can only help). The overall probability of failure is therefore $p$.
\end{proof}

\subsection{Warm start}
In order for the walk to converge quickly in each phase in total variation distance, we need the point from the previous phase will give a warm start for the current phase. In our algorithm, the $M$-warmness will be bounded by a constant. 

\begin{lemma}\label{lem:warmstart}
Let $f_i(x)$ be a function restricted to $K$ with density $e^{-a_i\|x\|^2}$ where $a_i = a_{i-1}(1-1/n)$, and $Q_i$ be a probability distribution proportional to $f_i(x)/\int_K f_i(x)dx$. Then, we have that 
\[
M(Q_i,Q_{i+1}) \leq \sqrt{e}.
\]
\end{lemma}

\begin{proof}
Let
\[
A = \frac{\int_K e^{-a_{i+1}\|x\|^2}dx}{\int_K e^{-a_i \|x\|^2}dx}.
\]
Then, 
\[
M(Q_i,Q_{i+1}) = \sup_{S \subseteq K} \frac{Q_i(S)}{Q_{i+1}(S)} \leq \sup_{x \in K} \frac{Q_i(x)}{Q_{i+1}(x)} = \sup_{x \in K} A \frac{e^{-a_i \|x\|^2}}{e^{-a_{i+1} \|x\|^2}}
\]
\[
= A\cdot \sup_{x \in K} e^{-a_i\|x\|^2/n}=A \text{ since } 0 \in K.
\]

We will now bound $A$. First, we extend $A$ to be over all $\mathbb{R}^n$ instead of $K$, and then argue that it can only decrease when restricted to $K$.
\begin{align*}
\frac{\int_{\mathbb{R}^n} e^{-a_{i+1}\|x\|^2}dx}{\int_{\mathbb{R}^n}e^{-a_i \|x\|^2}dx} &=  \frac{(a_{i}/\pi)^{n/2}}{(a_{i+1}/\pi)^{n/2}}\frac{(a_{i+1}/\pi)^{n/2}}{(a_{i}/\pi)^{n/2}}\frac{\int_{\mathbb{R}^n} e^{-a_{i+1}\|x\|^2}dx}{\int_{\mathbb{R}^n}e^{-a_i \|x\|^2}dx}\\
&=\frac{a_i^{n/2}}{a_{i+1}^{n/2}} = (\frac{1}{1-1/n})^{n/2} \leq \sqrt{e}.
\end{align*}

%\frac{S_n\int_0^\infty r^{n-1}e^{-a_{i+1}r^2}dr}{S_n \int_0^\infty r^{n-1}e^{-a_ir^2}dr} = \frac{a_{i+1}^{-n/2}\Gamma (n/2)/2}{a_i^{-n/2}\Gamma (n/2)/2} = \frac{a_i^{-n/2}(1-1/n)^{-n/2}}{a_i^{-n/2}} \leq \sqrt{e}.

Let $\mu_K(r)$ be the proportion of the sphere of radius $r$ centered at $0$ that is contained in $K$. Note that since $K$ is a convex body that contains $0$, that $r_1 > r_2 \Rightarrow \mu_{K}(r_1) \leq \mu_{K}(r_2)$. Then,
\[
A = \frac{\int_0^\infty r^{n-1} e^{-a_{i+1} r^2} \mu_K(r) dr}{\int_0^\infty r^{n-1} e^{-a_i r^2} \mu_K(r) dr}.
\]

Note $(r^{n-1}e^{-a_{i+1}r^2})/(r^{n-1}e^{-a_i r^2})$ is a monotonically increasing function with $r$. Since $K$ is a convex body containing $0$, we can partition $K$ into infinitesimally small cones centered at $0$. Consider an arbitrary cone $C$. $\mu_C(r)$ is $1$ for $r \in [0,r']$ and then $0$ for $r \in (r',\infty)$ since $K$ is convex. Since $(r^{n-1}e^{-a_{i+1}r^2})/(r^{n-1}e^{-a_i r^2})$ is monotonically increasing, the integral over the cone only gets larger by extending $\mu_C(r)$ to be $1$ for $r \in [0, \infty)$. Therefore
\[
\frac{\int_0^\infty r^{n-1} e^{-a_{i+1} r^2} \mu_C(r) dr}{\int_0^\infty r^{n-1} e^{-a_i r^2} \mu_C(r) dr} \leq \frac{\int_0^\infty r^{n-1} e^{-a_{i+1} r^2}dr}{\int_0^\infty r^{n-1} e^{-a_i r^2}dr} = \sqrt{e}.
\]
Since $C$ was an arbitrary cone from a partition of $A$, we have that $A \leq \sqrt{e}$.
\end{proof}

\subsection{Complexity of sampling}
We will now give a proof of Theorem \ref{thm:sampling}, which gives the complexity for getting a random point from the standard Gaussian intersected with the convex body $K$.\\

\begin{proof}(of Theorem \ref{thm:sampling})
First we note that using Cor. \ref{cor:iso-conc-bound}, the measure of $K\cap 4\sigma \sqrt{n}B_n$ is at least $1-e^{-n}$ of the measure of $K$.

Next, from Lemma \ref{lem:warmstart}, our algorithm will always have a warm start that is at most a constant. By Theorem \ref{thm:mixing}, alloting a failure probabilty to each sampler call of $O(\eps^2 p /n)$, we have that the number of steps per sampling phase is $O(n^2 \log(n/\eps) \log(n/\eps p))$, with an error parameter $\nu = O((\eps/n)^{15})$. From the algorithm description, the algorithm will have $n(\log (n) + \log \log(1/\eps))$ phases. Assuming that $\eps = \Omega(2^{-n})$, we have that the complexity of the first random point is $O(n^3 \log (n) \log (n/\eps)\log(n/\eps p))$. By Theorem \ref{thm:mixing}, the complexity for each subsequent random point from the standard Gaussian is $O(n^2 \log(1/\eps)\log(n/\eps p))$.
\end{proof}

\section{Analysis of volume algorithm}

In this section, we prove Theorem \ref{thm:volume}. This consists of the following parts:
\begin{enumerate}
\item The starting distribution is concentrated inside $K$.
\item The variance of the ratio estimated in each volume phase is bounded by a constant. 
\item The estimates obtained from multiple samples and across multiple phases are accurate with high probability. This requires bounding the errors due to the mild dependence between samples.
\end{enumerate}
We will establish these properties in subsequent sections, which combined with Theorem \ref{thm:sampling} and Theorem \ref{thm:mixing} complete the proof of the main theorem.

\subsection{Starting distribution}
The following lemma shows that the starting distribution of our algorithm will (1) have a volume that can be approximated by a standard integral computation and (2) can be efficiently sampled from using simple rejection sampling--generate a point from the distribution and reject if it is not in $K$.

\begin{lemma}
If $a \geq (n+\sqrt{8n \ln (1/\eps)})/2$ and $B_n \subseteq K$, then $\int_{K} e^{-a \|x\|^2}dx \geq (1-\eps)\int_{\mathbb{R}^n} e^{-a \|x\|^2}dx$.
\end{lemma}
\begin{proof}
We will use the following concentration bound on a spherical Gaussian in $\mathbb{R}^n$ with mean $\mu$ and variance $\sigma^2$, which is valid for $t>1$:

\[
\Pr(\|X-\mu\|^2-\sigma^2n > t \sigma^2 \sqrt{n}) \leq e^{-t^2/8}.
\]

Selecting $\mu = 0$, $t=\sqrt{8\ln(1/\eps)}$, and $\sigma^2 = 1/(n+t\sqrt{n})$ gives

\[
\Pr(\|X\|^2 > 1) \leq \eps
\]

and therefore all but an $\eps$-fraction of the Gaussian is contained inside $B_n$ (and therefore $K$). A Gaussian of the form $e^{-a\|x\|^2}$ has variance $\frac{1}{2a}$. Therefore,

\[
\frac{1}{2a} = \frac{1}{n+t\sqrt{n}} \Rightarrow a = \frac{n+\sqrt{8n \ln (1/\eps)}}{2}.
\]

\end{proof}

\subsection{Variance of ratio}
We need to bound the variance of the random variables, $W_i$ in the algorithm, used to estimate the ratios that yield the Gaussian volume. First, we will need the following lemma that is proved in \cite{LV2}:

\begin{lemma} \label{lem:z-logconcave}
Let $K \subseteq \mathbb{R}^n$ be a convex body and $f : K \rightarrow \mathbb{R}$ be a logconcave function. For any $a > 0$, define 

\[
Z(a) = \int_K f(ax) dx.
\]

\noindent
Then $a^nZ(a)$ is a logconcave function of $a$.

\end{lemma}

Using the above lemma, we can bound the variance of the ratio estimators in the main algorithm.

\begin{lemma}\label{lem:bdd-var}
Let $X$ be a random point in $K$ with density proportional to $e^{-a_i \|x\|^2}$, where $i$ is a multiple of $\lfloor \sqrt{n} \rfloor$, $j$ be the index of the distribution with which we are computing our ratio (i.e. $e^{-a_j \|x\|^2)}$), and $Y=e^{(a_j-a_i)\|X\|^2}$. Then, 
\[
\frac{\E(Y^2)}{\E(Y)^2} \leq (\frac{a_i^2}{a_j(2a_i-a_j)})^{n+1} < 8.
\]

\end{lemma}

\begin{proof}
We have
\[
\E(Y) = \frac{\int_K e^{-a_i \|x\|^2}dx}{\int_K e^{-a_j \|x\|^2}dx}
\]

and
\[
\E(Y^2) = \frac{\int_K e^{-(2a_i - a_j)\|x\|^2}dx}{\int_K e^{-a_j\|x\|^2}dx}.
\]

By Lemma \ref{lem:z-logconcave}, the function $a^{n+1} \int_K e^{-a\|x\|^2}dx$ is logconcave, and using  $f(x)f(y) \leq f((x+y)/2)^2$ for any logconcave function $f(\cdot)$ and $x,y \in \R^n$,
\[
\int_K e^{-a_j \|x\|^2}dx\int_K e^{-(2a_i - a_j)\|x\|^2} dx \leq (\frac{a_i^2}{a_j(2a_i - a_j)})^{n+1}(\int_K e^{-a_i\|x\|^2}dx)^2.
\]

In our algorithm, we have that
\[
a_i \geq a_j(1-\frac{1}{n})^{\lfloor \sqrt{n} \rfloor} \geq a_j (1-\frac{\lfloor \sqrt{n} \rfloor}{n}) \geq a_j (1-\frac{1}{\sqrt{n}}).
\]

Therefore, 
\[
\frac{\E(Y^2)}{\E(Y)^2} \leq (\frac{a_i^2}{a_j(2a_i-a_j)})^{n+1} \leq (1+\frac{1}{n-2\sqrt{n}})^{n+1},
\]

which implies the lemma.
\end{proof}

\subsection{Accuracy of computation}
We need to show that the answer computed by our algorithm is accurate by examining the error introduced at each phase and the dependence between sample points. We would like for all of our samples to be independent. However, each sample must have a warm start, which introduces a small dependence in our samples. We need to quantify this dependence and show how much it hurts our answer. We show that we are within the desired accuracy with probability $4/5$, and then allot $1/20$ failure probability to the sampling phases to obtain an overall probability of success of $3/4$. We can then use the standard trick to boost the probability of success to $1-\delta, \delta>0$ by repeating the experiment $O(\log(1/\delta))$ times and taking the median of the trials.

For two random variables $X,Y$, we will measure their independence by the following:
\[
\mu (X,Y) = \sup_{A,B} |P(X \in A, Y \in B) - P(X \in A)P(Y \in B)|,
\]
where $A,B$ range over measurable subsets of the ranges of $X,Y$.

We will give an argument similar to \cite{LV2}, and use the following lemmas that were proved there:

\begin{lemma}\label{lem:fn-indep}
If $f$ and $g$ are two measurable functions, then 
\[
\mu (f(X), g(Y)) \leq \mu (X,Y).
\]
\end{lemma}

\begin{lemma}\label{lem:cov-bd}
Let $X,Y$ be random variables such that $0 \leq X \leq a$ and $0 \leq Y \leq b$. Then
\[
|\E(XY)-\E(X)E(Y)| \leq ab\mu(X,Y).
\]
\end{lemma}

\begin{lemma}\label{lem:exp-bd}
Let $X\geq 0$ be a random variable, $a>0$, and $X' = \min (X,a)$. Then
\[
\E(X') \geq E(X) - \frac{\E(X^2)}{4a}.
\]
\end{lemma}

Now, we will show that the answer we get is within the correct bound with some probability.
\begin{lemma}
With probability at least $4/5$, 
\[
(1-\frac{\eps}{2})R_1 \ldots R_m \le W_1 \ldots W_m \le (1 + \frac{\eps}{2})R_1 \ldots R_m
\]
\end{lemma}
\begin{proof}
Let $t$ be the total number of sample points and $(X_0, X_1, X_2,  \ldots , X_t)$ be the total sequence of sample points taken by our algorithm. Then, some subsequence of these points will be the samples used the estimate the volume ratio, $(Z_0^0,Z_1^0, \ldots ,Z_k^0,Z_0^1,Z_1^1, \ldots ,Z_k^1, \ldots ,Z_0^m,Z_1^m, \ldots ,Z_k^m)$ where $Z_i^j$ is the $i$th sample point used for estimating the $j$th volume ratio. 

The distribution of each $X_i$ is approximately the correct distribution, but slightly off based on the error parameter, $\nu$, in each phase that bounded the total variation distance. We will define new random variables $\bar X_i$ that have the correct distribution for each phase.

Note that $X_0$ would be sampled from the exact distribution, and then rejected if outside of $K$. Therefore, $X_0 = \bar X_0$, and $P(X_1 = \bar X_1) \geq 1-\nu$ from the definition of total variation distance. So, by induction, we see that 
\begin{equation}\label{eqn:exact-chance}
\Pr(X_i = \bar X_i, \text{ }  \forall i \in [0,t]) \geq 1-t\nu.
\end{equation}

Let $\bar Z_j^i$ denote a random variable from the correct distribution for $Z_j^i$. Let $z_i = i\lfloor \sqrt{n} \rfloor$, and let 
\[
Y_j^i = e^{(a_{z_i}-a_{z_{i+1})}\| \bar Z_j^i \|^2}.
\]

Then, let 
\[
\bar W_i = \frac{1}{k}\sum_{j=1}^k Y_j^i.
\]

From \cite{LV2}, for every $i$, we have that $\E(\bar W_i) = \E(\bar Y_j^i) = R_i$, and by Lemma \ref{lem:bdd-var}, $E((Y_j^i)^2) \leq 8E(Y_j^i)^2$. Thus,
\begin{equation}\label{eqn:bdd-var}
\E(\bar W_i^2) = \frac{1}{k^2} (\sum_{j=1}^k \E((Y_j^i)^2) + k(k-1)R_i^2) \leq (1+\frac{7}{k})R_i^2.
\end{equation}

If we had independence between our samples, we could use (\ref{eqn:bdd-var}) to prove the lemma. However, since we use previous points as the starting points for the next, the sample points are not independent. We use the following lemma from \cite{LV2} to bound the dependence between our random variables.

\begin{lemma}\label{lem:delta-ind}

(a) For $0 \leq i < t$, the random variables $X_i$ and $X_{i+1}$ are $\nu$-independent, and the random variables $\bar{X}_i$ and $\bar{X}_{i+1}$ are $(3 \nu)$-independent.

(b) For $0 \leq i < t$, the random variables $(X_0, \ldots ,X_i)$ and $X_{i+1}$ are $(3\nu)$-independent.

(c) For $0 \leq i < m$, the random variables $\bar{W}_1 \ldots \bar{W}_i$ and $\bar{W}_{i+1}$ are $(3km\nu)$-independent.
\end{lemma}

The variables $\bar W_i$ are not bounded, but we will introduce a new set of random variables based on $\bar W_i$ that are bounded so we can later apply Lemma \ref{lem:cov-bd}. Let
\[
\alpha = \frac{\eps^{1/2}}{8(m\mu)^{1/4}},
\]
where $\mu = 3km\nu$. Note that $\alpha$ is much larger than one. Define
\[
V_i = \min\{\bar W_i, \alpha \E (\bar W_i) \}.
\]

It is clear that $\E(V_i) \leq \E(\bar W_i)$, and by Lemma \ref{lem:exp-bd}, we also have
\[
\E(V_i) \geq \E(\bar W_i) - \frac{\E(\bar W_i^2)}{4\alpha \E(\bar W_i)} \geq (1 - \frac{1}{4\alpha}(1 + \frac{7}{i})^m)\E(\bar W_i) \geq (1-\frac{1}{2\alpha})\E(\bar W_i).
\]

Let $U_0=1$ and define recursively
\[
U_{i+1} = \min\{U_iV_{i+1}, \alpha \E(V_1) \ldots \E(V_{i+1})\}.
\]

We will now show that 
\begin{equation}\label{eqn:u-bound}
(1-\frac{i-1}{\alpha})\E(V_1) \ldots \E(V_i) \leq \E(U_i) \leq (1 + 2\mu \alpha^2 i)\E(V_1) \ldots \E(V_i).
\end{equation}
By Lemma \ref{lem:fn-indep}, the random variables $U_i$ and $V_{i+1}$ are $\mu$-independent, and by Lemma \ref{lem:cov-bd} and since $\alpha \geq 1$, 
\begin{equation}\label{eqn:uv-bound}
|\E(U_iV_{i+1} - \E(U_i)\E(V_{i+1})| \leq \mu \alpha \E(V_1) \ldots \E(V_i)\alpha \E(\bar W_{i+1}) \leq 2\mu\alpha^2 \E(V_1) \ldots \E(V_{i+1}).
\end{equation}

From (\ref{eqn:uv-bound}), we can get the upper bound on $\E(U_{i+1})$ by induction:
\begin{align*}
\E(U_{i+1}) &\leq \E(U_i V_{i+1}) \leq \E(U_i)\E(V_{i+1}) + 2 \mu \alpha^2 \E(V_1) \ldots \E(V_{i+1})\\
&\leq (1+2\mu \alpha^2(i+1))\E(V_1) \ldots \E(V_{i+1}). \numberthis \label{eqn:u-upperbound}
\end{align*}

Similarly, 
\begin{align}
&\E(U_i^2) \leq (1+2\mu \alpha^4 i) \E(V_1^2) \ldots \E(V_i^2)\label{eqn:u2-bound}\\
\text{and } \>\> &\E(U_i^2V_{i+1}^2) \leq (1 + 2\mu \alpha^4 i)\E(V_1^2) \ldots \E(V_{i+1}^2).\label{eqn:uv2-bound}
\end{align}

For the lower bound, we use Lemma \ref{lem:exp-bd} and (\ref{eqn:uv2-bound}) to get:
\begin{align*}
\E(U_{i+1}) &\geq \E(U_iV_{i+1}) - \frac{\E(U_i^2V_{i+1}^2)}{4\alpha \E(V_1) \ldots \E(V_{i+1})}\\
&\geq \E(U_iV_{i+1}) - (1+2\mu \alpha^4 i)\frac{\E(V_1^2) \ldots \E(V_{i+1}^2)}{4\alpha \E(V_1) \ldots \E(V_{i+1})}. \numberthis \label{eqn:u-lowerbound}
\end{align*}

For $\alpha \geq 3k$, we have that
\begin{align*}
\E(V_i^2) \leq \E(\bar{W}_i^2) &\leq (1+\frac{7}{k})\E(\bar{W}_i)^2\\
&\leq (1+\frac{7}{k})\frac{1}{(1-1/(2\alpha))^2}\E(V_i)^2\\
&\leq (1+\frac{7}{k})(1+\frac{1}{2k})\E(V_i)^2\\
&\leq (1+\frac{8}{k})\E(V_i)^2. \numberthis \label{eqn:v2-bound}
\end{align*}

Combining (\ref{eqn:uv-bound}), (\ref{eqn:u-lowerbound}), and (\ref{eqn:v2-bound}), 
\begin{align*}
\E(U_{i+1}) &\geq \E(U_i V_{i+1}) - \frac{1}{4\alpha}(1+2\mu\alpha^4i)(1+\frac{8}{k})^i\E(V_1) \ldots \E(V_{i+1})\\
&\geq \E(U_iV_{i+1})-\frac{1+2\mu\alpha^4i}{2\alpha}\E(V_1) \ldots \E(V_{i+1})\\
&\geq \E(U_i)\E(V_{i+1}) - \frac{1}{\alpha}\E(V_1) \ldots \E(V_{i+1}).
\end{align*}

Then, by induction on $i$,
\begin{equation}\label{eqn:u-lowerbound2}
\E(U_{i+1}) \geq \E(V_1) \ldots \E(V_{i+1}) - \frac{i}{\alpha}\E(V_1) \ldots \E(V_{i+1}).
\end{equation}

Putting (\ref{eqn:u-upperbound}) and (\ref{eqn:u-lowerbound2}) together, we now have a proof of (\ref{eqn:u-bound}). Thus,
\[
\E(U_m) \leq (1 + \frac{\eps}{4}) \E(V_1) \ldots \E(V_m) \leq (1 + \frac{\eps}{4})\E(\bar{W}_1) \ldots \E(\bar{W}_m).
\]

We also have that $\alpha \geq 4m/\eps$ implies 
\[
\E(U_m) \geq (1-\frac{\eps}{4})\E(\bar{W}_1) \ldots \E(\bar{W}_m).
\]

From (\ref{eqn:u2-bound}) and (\ref{eqn:v2-bound}), and the selection of $\alpha$, $\mu$, and the lower bound on $k$, we have that 
\begin{align*}
\E(U_m^2) &\leq (1+2\mu \alpha^4 m)\E(V_1^2) \ldots \E(V_m^2)\\
&\leq (1+2\mu \alpha^4 m)(1+\frac{8}{k})^m\E(V_1)^2 \ldots \E(V_m)^2\\
&\leq (1+2\mu \alpha^4 m)(1+\frac{8}{k})^m\frac{1}{(1-(m-1)/\alpha)^2}\E(U_m)^2\\
&\leq (1 + \frac{\eps^2}{64})E(U_m)^2,
\end{align*}

and hence

\[
\Pr\big (|U_m - \E(U_m)| \leq \frac{\eps}{2}\E(\bar{W}_1) \ldots \E(\bar{W}_m)\big ) \geq 0.9
\]
 by Chebyshev's inequality. Then, applying Markov's inequality,
\[
\Pr(U_{i+1} \neq U_iV_{i+1}) = \Pr \big (U_iV_{i+1} > \alpha \E(V_1) \ldots \E(V_{i+1})\big )\leq\frac{2}{\alpha}
\]
and similarly
\[
\Pr(V_i \neq \bar{W}_i)\leq \frac{1}{\alpha}.
\]

So, with probability at least $1-3k/\alpha$, we have $U_m = \bar{W}_1 \ldots \bar{W}_m$. Also, from (\ref{eqn:exact-chance}), we have that $\bar{W}_1 \ldots \bar{W}_m = W_1 \ldots W_m$ with probability at least $1 - 2km\nu$. Note that $\E(\bar{W}_1) \ldots \E(\bar{W}_m) = R_1 \ldots R_m$. Therefore, with probability at least $4/5$
\begin{align*}
|W_1 \ldots W_m - R_1\ldots R_m| \leq \frac{\eps}{2}R_1 \ldots R_m,
\end{align*}
which proves the lemma.
\end{proof}

\section{Concluding remarks}
\begin{enumerate}
\item The focus of this paper was integrating a Gaussian over the convex region, but an important problem in statistics and probability is integrating a function $f$ with respect to the Gaussian density. For logconcave $f$, the isoperimetric inequality of Theorem \ref{thm:iso} still applies. We expect that the rest of our analysis can also be extended to this setting. 
\item For a nonspherical Gaussian, the algorithm complexity should grow as $O^*(\Lambda^2 n^3)$, where $\Lambda^2$ is the largest variance of the Gaussian in any direction.
\item Can the upper bound for hit-and-run be improved for Gaussian sampling? An important ingredient that is missing is an analog of Theorem \ref{thm:iso} for the cross-ratio distance $d_K$ relevant to hit-and-run. An analysis of hit-and-run could have two advantages: rapid mixing from any starting distribution, and in the Gaussian volume algorithm, there would be no need for interleaving of sampling and volume phases, thus making the remaining analysis simpler, as in \cite{LV2}.
\end{enumerate}

\noindent
{\bf Acknowledgements.} We are grateful to Ronen Eldan and Mario Ullrich for helpful comments on an earlier draft.

\bibliography{acg}

\end{document}